\documentclass[12pt,draft,onecolumn]{IEEEtran}
\usepackage{amsfonts,amsmath,graphicx,cite,bm,amssymb,amsthm,enumerate,epsfig,psfrag,cases,mathtools}
\renewcommand{\(}{\left(}
\renewcommand{\)}{\right)}
\renewcommand{\[}{\left[}
\renewcommand{\]}{\right]}

\newcommand{\Tr}[1]{{\rm{Tr}}\left(#1\right)}

\newcommand{\End}[1]{{\rm{End}}}

\renewcommand{\log}[1]{{\rm{log}}#1}

\newcommand{\diag}[1]{{\rm{diag}}\left\{#1\right\}}

\renewcommand{\vec}[1]{{\rm{vec}}\(#1\)}

\newtheorem{lemma}{Lemma}
\newtheorem{definition}{Definition}
\newtheorem{theorem}{Theorem}

\newtheorem{corollary}{Corollary}

\newtheorem{assm}{Assumption}
\newtheorem{rem}{Remark}

\newcommand{\norm}[1]{\left\lVert#1\right\rVert}

\usepackage{fontenc}
\usepackage{inputenc}
\usepackage[square,sort,comma,numbers]{natbib}
\usepackage[table]{xcolor}
\usepackage[ruled,vlined]{algorithm2e}
\usepackage{mathtools}
\DeclarePairedDelimiter{\ceil}{\lceil}{\rceil}

\begin{document}
\title{Error Bound for Compound Wishart Matrices}

\author{Ilya Soloveychik, \\ the Hebrew University of Jerusalem, Israel
\thanks{This work was partially supported Kaete Klausner Scholarship, the Hebrew University of Jerusalem.}}
\maketitle

\begin{abstract}
In this paper we consider non-asymptotic behavior of the real compound Wishart matrices that generalize the classical real Wishart distribution. In particular, we consider matrices of the form $\frac{1}{n}XBX^T$, where $X$ is a $p \times n$ matrix with centered Gaussian elements and $B$ is an arbitrary $n \times n$ matrix and sequences of such matrices for varying $n$. We show how the expectation of deviations from the mean can be bounded for compound Wishart matrices.
\end{abstract}


\IEEEpeerreviewmaketitle

\section{Introduction}
For a real $p \times n$ matrix $X$ consisting of independent standard normally distributed elements the Wishart matrix, defined as $W = \frac{1}{n}XX^T$, was introduced by Wishart \cite{wishart1928generalised}, who also derived the law of its distribution. As evidenced by the wide interest among the scientists and engineers, the Wishart law is of primary importance to statistics, see e.g. \cite{anderson1958introduction, muirhead2009aspects}. In particular, Wishart law describes exactly the distribution of the sample covariance matrix in the Gaussian populations. As a natural generalization, the compound Wishart matrices were introduced by Speicher \cite{speicher1998combinatorial}.
\begin{definition} (Compound Wishart Matrix)
Let $X_i \sim \mathcal{N}(0,\Theta),\; i=1,\dots,n$, where $\Theta$ is $p \times p$ real positive definite matrix, and $B$ be an arbitrary real $n \times n$ matrix. We say that a random $p \times p$ matrix $W$ is Wishart with shape parameter $B$ and scale parameter $\Theta$ if
\begin{equation}
W = \frac{1}{n}X B X^T,
\label{wish_def}
\end{equation}
where $X = [X_1,\dots,X_n]$. We write $W \sim \mathcal{W}(\Theta, B)$.
\end{definition}
\begin{rem}
\label{rem}
It will sometimes be instructive to use the representation $W = \frac{1}{n}\Theta^{1/2} Y B X^T \Theta^{1/2}$, where $Y = [Y_1,\dots,Y_n]$ and $Y_i \sim \mathcal{N}(0,I)$.
\end{rem}

When B is positive definite, \cite{burda2011applying} interprets W as a sample covariance under correlated sampling. Similar interpretation for the complex case appears in \cite{chuah2002capacity}. The usual real Wishart matrices correspond to the choice $B = I$. Wishart distribution and its generalization to the case of positive definite matrix $B$ are widely used in economics, in particular in portfolio allocation theory \cite{collins2013compound}. 

When applying the representation theoretic machinery to the study of symmetries of normally distributed random vectors (see \cite{shah2012group} for an example of such settings) we encountered the case of compound Wishart distribution with a skew symmetric $B$ of the form (\ref{sk_sym}) and this problem motivated the current work.

Most of the literature concerning Wishart distribution deals with the asymptotic, $n \rightarrow \infty$, and the double asymptotic $n, p \rightarrow \infty$ regimes. In particular, a version of Marchenko-Pastur law was generalized to this case, see \cite{bryc2008compound} and references therein for a wide survey on the asymptotic behavior of the compound Wishat matrices.

In a different line of research in the recent years there were observed a few prominent achievements in the random matrix theory concerning the non-asymptotic regime. In particular, Bernstein inequality and other concentration inequalities \cite{tropp2012user, mackey2012matrix} were generalized to the matrix case using the Stein's method of exchangeable pairs and Lieb theorem. The finite sample performance of the sample covariance matrix was also profoundly investigated for a large class of distributions, see \cite{vershynin2012close, srivastava2011covariance, eldar2012compressed} and references therein. 

Closely related to the Wishart family of distributions is partial estimation of covariance matrices by the sample covariance in Gaussian models. In particular, matrices of the form $\frac{1}{n}M\cdot X X^T$, where $"\cdot"$ denotes the Hadamard (entry-wise) product of matrices, became attractive due to the advances in structured covariance estimation \cite{bickel2008covariance, bickel2008regularized, rothman2009generalized, cai2010optimal}. The matrix $M$ represents the a priory knowledge about the structure of the true covariance matrix in the form of a mask. The most widespread examples of such assumptions are banding, tapering and thresholding, which assume the elements of the mask $M$ belong to the interval $[0,1]$. The non-asymptotic behavior of such masked sample covariance matrices was investigated in \cite{levina2012partial, chen2012masked}. It has been shown that the sample complexity (the minimal number of samples need to achieve some predefined accuracy with stated probability) is proportional to $m \log^c{(2p)}$, where $m$ is a sparsity parameter of $M$ and $c \geq 1$. 

The main purpose of this paper is to demonstrate the non-asymptotic results concerning the average deviations of the compound Wishart matrices from their mean. We actually consider two closely related problems. 
\begin{itemize}
\item \text{\bf{Single Wishart matrix}} The first one can be roughly described as following: given the particular values of dimension $p$, number of samples $n$ and shape matrix $B$, what precision we can obtain when approximating the expected value of the corresponding Wishart matrix.
\item \text{\bf{Sequence of Wishart matrices}} Another setting arise from a different approach to the problem: assume we are given a sequence of matrices $B_n$ satisfying some assumptions stated below and the dimension $p$ is fixed. This data provides us with a sequence of Wishart matrices $W_n = \frac{1}{n}XB_nX^T$, where the number of columns in $X$ changes appropriately (an exact definition provided below). Assume in addition that all these Wishart matrices have a common expectation $W^0$. The natural question is how many measurements $n$ does one need to collect in order to estimate the mean value $W^0$ accurately.
\end{itemize}
Both of these problems arise in different areas of research and non-asymptotic analysis is often required. The problem becomes especially critical when the values $p$ and $n$ are of the same order of magnitude. Below we provide a theorem answering the two posed questions. Although the result obtained in Corollary \ref{main_th_2} is related to the case of fixed dimension $p$, it can be extended to a sequence $W_n^p$, where the dimension $p=p(n)$ varies with $n$, while keeping the spectral properties of the sequence of corresponding covariance matrices $\Theta_n$ controlled.

In particular, a partial answer to the second problem can be formulated as following: the number of samples proportional to $\sqrt{p} \ln^2{p}$ is needed to accurately estimate the expectation of the compound Wishart matrix.

The rest of the text is organized as following. After the notations section we provide additional definitions and the statements of the results. Then a few examples demonstrating the applications are given. The proof of the theorem concludes the paper.

\subsection*{Notations}
Capital last letters of the English alphabet ($W,X,Y,Z$) denote random entities, all the other letters stand for deterministic entities. 
For an arbitrary rectangular matrix $A,\: \norm{A}$ denotes its spectral norm and $\norm{A}_{\text{Frob}} = \sqrt{\Tr{AA^T}}$ stands for its Frobenius (trace) norm. For two vectors $v, u$ laying in a Euclidean space, $(u,v)$ denotes their scalar product and $\norm{v}_2$ the corresponding length. $I_m$ stands for the $m \times m$ identity matrix, when the dimension is obvious from the context the subscript is omitted.

\section{Problem Formulation and the Main Results}
In addition to the definitions given above we define a notion of sequence of Wishart matrices, corresponding to the second question stated in the Introduction.

\begin{definition}(Sequence of Wishart Matrices)
Consider a sequence $\{B_n\}_{n\in\mathcal{S}}$ of real $n \times n$ deterministic matrices, where $\mathcal{S} \subset \mathbb{N}$ is ordered, for every $n \in \mathcal{S}$ let $X_i \sim \mathcal{N}(0,\Theta), i=1,\dots,n$, where $\Theta$ is $p \times p$ real positive definite matrix, then define the sequence of Wishart matrices as
\begin{equation}
W_n = \frac{1}{n}X B_n X^T,
\label{wish_def}
\end{equation}
where $X = [X_1,\dots,X_n]$.
\end{definition}
The same Remark \ref{rem} as above applies here, as well. Note also that the dimension of $X$ depends on $n$, but this is not reflected by an additional subscript.

To make this definition useful and meaningful we will have to make some assumptions on the sequence $\{W_n\}_{n\in\mathcal{S}}$. In particular, we first want to answer the following question: what properties should we require from the sequence $\{B_n\}_{n\in\mathcal{S}}$ to make the sequence $\{W_n\}_{n\in\mathcal{S}}$ interesting to investigate. Below we fix the dimension $p$ and refer to vectors $X_i, i=1,\dots,n$ as measurements. So what actually changes from matrix to another in the sequence $W_n$ is the underlying matrix $B_n$ and the respective number of measurements.

The examples of sequences $\{B_n\}_{n\in\mathcal{S}}$ are the following:
\begin{itemize}
\item The most widely used is the sequence of diagonal matrices $B_n = \diag{b_1,\dots,b_n}, n \in \mathcal{S}$. When $\Tr{B}=n$ the expectation of the Wishart sequence coincides with the covariance matrix $\Theta$ as shown below.
\item Another common example is a sequence of skew-symmetric matrices of the form
\begin{equation}
B_n =
 \begin{pmatrix}
  0 & I_{n/2}  \\
  -I_{n/2} & 0
  \end{pmatrix}, 
\label{sk_sym}
\end{equation}
where $n$ is assumed even. We encountered this case when investigating the group symmetry properties of sample covariance matrices.  
\end{itemize}

In order to generalize the properties of $\{B_n\}_{n\in\mathcal{S}}$ we encountered in the application we state an additional auxiliary result that we did not found in the literature. 
\begin{lemma}
\label{consist_l}
Let $B$ be a real $n \times n$ matrix and $X$ real $p \times n$ with independent standard normally distributed elements, then for $W = \frac{1}{n}XBX^T$
\begin{equation}
\mathbb{E}(W) = \frac{\Tr{B}}{n}I.
\end{equation}
\begin{proof}
Denote the expectation $\mathbb{E}(W)$ by $W^0$ and consider the elements of $W^0$:
\begin{equation}
W^0_{ij} =  \frac{1}{n}\mathbb{E}\(\sum_{k,l=1}^n X_{ik} B_{kl} X_{jl}\).
\end{equation}
As all $X_{ik},i=1,\dots,p,\; k=1,\dots,n$ are independent we get immediately that
\begin{equation}
W^0_{ij} = \frac{1}{n}\sum_{k,l=1}^n B_{kl} \mathbb{E} \(X_{ik} X_{jl}\) = 0,\quad i \neq j,
\end{equation}
\begin{equation}
W^0_{ii} = \frac{1}{n}\sum_{k,l=1}^n B_{kl} \mathbb{E} \(X_{ik}X_{il}\) = \frac{1}{n}\sum_{k}^n B_{kk} \mathbb{E} \(X_{ik}^2\) =\frac{1}{n}\Tr{B}.
\end{equation}
And the statement follows.
\end{proof}
\end{lemma}

\begin{corollary}
Let $B$ be a real $n \times n$ matrix and $X_i \sim \mathcal{N}(0,\Theta), i=1,\dots,n$, where $\Theta$ is $p \times p$ real positive definite matrix, be independent, then for $X = [X_1,\dots,X_n],\; W = \frac{1}{n}XBX^T$
\begin{equation}
\mathbb{E}(W) = \frac{\Tr{B}}{n}\Theta.
\end{equation}
\end{corollary}

In particular, Lemma \ref{consist_l} implies that if $\mathcal{S}$ is unbounded, then to ensure the sequence $\{W_n\}_{n \in \mathcal{S}}$ is consistent we should at least demand $\frac{1}{n}\Tr{B_n} \rightarrow \beta \in \mathbb{R}$. We actually make a stronger
\begin{assm}
\label{assm_1}
The scaled traces $\frac{1}{n}\Tr{B_n}$ are all equal: $\frac{1}{n}\Tr{B_n} = \beta$ for all $n \in \mathcal{S}$. Since we can scale the sequence $\{W_n\}_{n \in \mathcal{S}}$, without loss of generality assume $\frac{1}{n}\Tr{B_n} = 1, \forall n \in \mathcal{S}$.
\end{assm}

\subsection*{Main results}
\begin{theorem} 
\label{main_th_1}
Let $\Theta$ is $p \times p$ real positive definite matrix and $X_i \sim \mathcal{N}(0,\Theta), i=1,\dots,n$, be independent. Let $B$ be an arbitrary real $n \times n$ matrix and denote $\kappa = \frac{\norm{B}_{\text{Frob}}}{\norm{B}},\; \sigma = \norm{B}$, then
\begin{equation*}
\mathbb{E}\norm{W - W^0} \leq \frac{24 \ceil{\ln{2p}}^2 \sqrt{p}(4\sigma+\kappa\sqrt{\pi})}{n}\norm{\Theta}.
\end{equation*}
\end{theorem}

\begin{corollary} 
\label{main_th_2}
Let $\Theta$ is $p \times p$ real positive definite matrix and $\{B_n\}_{n\in\mathcal{S}} \in \mathbb{R}^{n \times n}$, where $\mathcal{S} \subset \mathbb{N}$ is ordered. For every $n \in \mathcal{S}$ let $X_i \sim \mathcal{N}(0,\Theta), i=1,\dots,n$, be independent. Assume that $\Tr{B_n} = n,\; \forall n \in \mathcal{S}$, and denote $\kappa = \max_{n \in \mathcal{S}}\norm{B}_{\text{Frob}}$ and $\sigma = \max_{n \in \mathcal{S}}\norm{B_n}$, then
\begin{equation*}
\mathbb{E}\norm{W_n - W^0} \leq \frac{24 \ceil{\ln{2p}}^2 \sqrt{p}(4\sigma+\kappa\sqrt{\pi})}{n}\norm{\Theta}.
\end{equation*}
\end{corollary}

\section{Preliminaries}
\subsection{Proof outline} In the rest of this paper we prove Theorems \ref{main_th_1} and \ref{main_th_2}. We shall observe that for a Wishart matrix $W$ the quadratic form $(Wx,y)$ is a Gaussian chaos (defined below) for fixed unit vectors $x$ and $y$ on the sphere $S^{p-1}$. We control the chaos uniformly for all $x, y$ by establishing concentration inequalities depending on the "sparsity" of $x, y$. We do so using the techniques of decoupling, conditioning, and applying concentration bounds for Gaussian measure. After this we make use of covering arguments to measure the number of sparse vectors $x, y$ on the sphere. The general layout of the proof goes parallel to the proof given by \cite{levina2012partial}, we modify and generalize a few of their intermediate results to the case of non-symmetric matrices.

\subsection{Decoupling} We start by considering bilinear forms in normally distributed vectors. The following definition will be useful is the sequel.
\begin{definition}
Let $Z \in \mathbb{R}^p$ be a centered Gaussian random vector and $B$ a square $p \times p$ matrix, then the bilinear form $(BZ,Z)$ is called a quadratic Gaussian chaos.
\end{definition}
We generalize here Lemma 3.2 from \cite{levina2012partial} to the case of non-symmetric matrices.

\begin{lemma}(Decoupling of Gaussian Chaos)
\label{gaus_chaos}
Let $Z \in \mathbb{R}^p$ be centered normal random vector and $Z'$ its independent copy. Let also $\mathcal{B}$ be a subset of $p \times p$ square matrices. Then
\begin{equation*}
\mathbb{E}\sup_{B \in \mathcal{B}}|(B Z,Z) - \mathbb{E}(B Z,Z)| \leq \mathbb{E}\sup_{B \in \mathcal{B}}|(B Z,Z').
\end{equation*}
\end{lemma}
\begin{proof}
Without loss of generality assume that $Z$ is standard, otherwise plug $\Theta^{-1/2}B\Theta^{-1/2}$ instead of $B$ and follow the same reasoning (here $\Theta$ is the covariance matrix of $Z$).
\begin{align*}
&E:=\mathbb{E}_{Z}\sup_{B \in \mathcal{B}}|(B Z,Z) - \mathbb{E}(B Z,Z)| = \mathbb{E}_{Z}\sup_{B \in \mathcal{B}}|(B Z,Z) - \mathbb{E}_{Z'}(B Z',Z')| \\
&\leq \mathbb{E}_{Z,Z'}\sup_{B \in \mathcal{B}}|(B Z,Z) - (B Z',Z')|,
\end{align*}
where the equality is due to the fact that the distributions of $Z$ and $Z'$ are identical and the inequality is due to Jensen. In the calculation above we emphasized explicitly the variables of integration in the  expectations to make the transitions clear. For an arbitrary $B$ note the identity
\begin{equation*}
(BZ,Z) - (BZ',Z') = \(B\frac{Z+Z'}{\sqrt{2}},\frac{Z-Z'}{\sqrt{2}}\) + \(B\frac{Z-Z'}{\sqrt{2}},\frac{Z+Z'}{\sqrt{2}}\).
\end{equation*}
By rotation invariance of the standard Gaussian measure, the pair $\(\frac{Z+Z'}{\sqrt{2}},\frac{Z-Z'}{\sqrt{2}}\)$ is distributed identically with $(Z,Z')$, hence we conclude that
\begin{align*}
&E \leq \mathbb{E}_{Z,Z'}\sup_{B \in \mathcal{B}}|(B Z,Z') + (B Z',Z)| \leq \mathbb{E}_{Z,Z'}\sup_{B \in \mathcal{A}}|(B Z,Z')| + \mathbb{E}_{Z,Z'}\sup_{B \in \mathcal{B}}|(B Z',Z)| \\
&= 2\mathbb{E}_{Z,Z'}\sup_{B \in \mathcal{B}}|(B Z',Z)|,
\end{align*}
and the statement follows.
\end{proof}

\begin{lemma}
\label{decoup_bound}
Let $X_1,\dots,X_n,X'_1,\dots,X'_n \sim \mathcal{N}(0,\Theta)$, where $\Theta$ is a $p \times p$ real positive definite matrix, be all independent. Consider the compound Wishart matrix and its decoupled counterpart defined as
\begin{equation*}
W = \frac{1}{n} X BX^T, \quad W' = \frac{1}{n} X' BX^T.
\end{equation*}
Denote $W^0 = \mathbb{E}(W) = \frac{\Tr{B}}{n}\Theta$, then
\begin{equation*}
\mathbb{E}\norm{W - W^0} \leq 2\mathbb{E}\norm{W'}.
\end{equation*}
\end{lemma}
\begin{proof}
Using the definition of spectral norm we obtain
\begin{align*}
\mathbb{E}\norm{W - W^0} =\mathbb{E}\sup_{x,y \in S^{p-1}}|(W x,y) - \mathbb{E}(W x,y)|.
\end{align*}
We rewrite the inner product as
\begin{align*}
(W x,y) = \frac{1}{n}\sum_{i,j=1}^p \(\sum_{l,m=1}^n b_{lm}X_{il}X_{jm}\) x_i y_j = 
\frac{1}{n}\sum_{i,j=1}^p \sum_{l,m=1}^n b_{lm}X_{il}X_{jm} x_i y_j.
\end{align*}
Let us now stack the vectors $(X_1,\dots,X_n)$ into one long vector $\vec{X} := \vec{\{X_{il}\}_{i,l=1}^{p,n}}$ which is a normal vector of dimension $n \times p$, then it is easy to see that the right-hand side of the last equality is a quadratic Gaussian chaos in $\vec{X}$ and the previous lemma applies with the appropriate choice of $\mathcal{B}$.
\end{proof}

\subsection{Concentration}

\begin{lemma}
\label{st_dev}
Let $Z = (Z_1,\dots,Z_p) \sim \mathcal{N}(0,\Theta)$, $\Theta$ is a $p \times p$ positive definite matrix and let $a=(a_1,\dots,a_p) \in \mathbb{R}^p$. Then $\sum_{i=1}^pa_i Z_i$ is a centered normal variable with standard deviation $\norm{\Theta^{1/2}a}_2 \leq \norm{\Theta^{1/2}}\norm{a}_2$.
\end{lemma}

We state here an auxiliary result from Concentration of the Gaussian Measure theory. Such concentration results are usually stated in terms of the standard normal distribution, but they can be easily generalized for an arbitrary normal distribution as following
\begin{lemma} \cite{ledoux1991probability}
\label{conc_gauss_space}
Let $f: \mathbb{R}^p \to \mathbb{R}$ be a Lipschitz function with respect to the Euclidean metric with constant $L = \norm{f}_{\text{Lip}}$. Let $Z \sim \mathcal{N}(0,\Theta)$, $\Theta$ is a $p \times p$ positive definite matrix then
\begin{equation*}
\mathbb{P}(f(Z) - \mathbb{E}f(Z))\leq \frac{1}{2}\exp\(-\frac{t^2}{2L^2\norm{\Theta}}\),\quad \forall t\geq 0.
\end{equation*}
\end{lemma}

\subsection{Discretization}
Recall that in the Euclidean $p$ dimensional space the spectral norm of a square $p \times p$ (not necessarily symmetric) matrix $A$ can be defined as
\begin{align*}
\mathbb{E}\norm{A} =\mathbb{E}\sup_{x,y \in S^{p-1}}|(A x,y)|.
\end{align*}
We approximate the spectral norm of matrices by using $\varepsilon$-nets in the following way:
\begin{lemma} \cite{eldar2012compressed}
Let $A$ be a $p \times p$ matrix and $\mathcal{N}$ be a $\delta$-net of the sphere $S^{p-1}$ in the Euclidean space for some $\delta \in [1,0)$. Then
\begin{equation*}
\norm{A} \leq \frac{1}{(1-\delta)^2}\max_{x,y \in \mathcal{N}}(Ax,y).
\end{equation*}
\end{lemma}

Following \cite{levina2012partial}, we introduce the notion of coordinate-wise sparse \textit{regular vectors}.
\begin{definition} 
The subset of regular vectors of sphere $S^{p-1}$ is defined as
\begin{equation*}
\text{Reg}_p(s) = \{x \in S^{p-1}\mid\text{ all coordinates satisfy: }x_i^2 \in \{0,1/s\}\}, \quad s \in [p]\},
\end{equation*}
\begin{equation*}
\text{Reg}_p = \bigcup_{s \in [p]}\text{Reg}_p(s).
\end{equation*}
\end{definition}
\begin{lemma} \cite{levina2012partial}
\label{norm_bound}
Let $A$ be a $p \times p$ matrix, then
\begin{equation*}
\norm{A} \leq 12 \ceil{\ln{2p}}^2 \max_{x,y \in \text{Reg}_p}(Ax,y).
\end{equation*}
\begin{proof}
The proof can be found in \cite{levina2012partial}, it uses the regular vectors to construct a specific $\delta$-net and obtain the bound given in the statement.
\end{proof}
\end{lemma}

\section{Proof of Theorems \ref{main_th_1}, \ref{main_th_2}}
We partition the proof into a few sections.
\subsection{Decoupling and Conditioning}
Using Remark \ref{rem} we can rescale the random vectors and assume without loss of generality that $\Theta=I$. Now by Lemma \ref{decoup_bound} it suffices to estimate $\mathbb{E}\norm{W'}$. From Lemma \ref{norm_bound} we get that
\begin{equation}
\mathbb{P}(\norm{W'} \geq t) \leq \mathbb{P}(12 \ceil{\ln{2p}}^2 \max_{x,y \in \text{Reg}_p}(W'x,y) \geq t).
\label{norm_bo}
\end{equation}
Write the inner product coordinate-wise and rearrange the summands to obtain
\begin{align}
(W' x,y) = \frac{1}{n}\sum_{l=1}^n \sum_{i=1}^p \[\sum_{m=1}^n \sum_{j=1}^p b_{lm}x_j X_{jm}\] y_i X'_{il}.
\label{gaus_expr}
\end{align}
We now fix $x$ and $y$ and condition on the variables $X_{jm}, j=1,\dots,p,m=1,\dots,n$ so that the expression (\ref{gaus_expr}) defines a centered normal random variable. We wish to estimate its standard deviation with the help of Lemma \ref{st_dev}. Since we have assumed $\Theta=I$, the covariance matrix of the concatenated vector $\vec{X'}$ is also the identity matrix. Then Lemma \ref{st_dev} implies that (\ref{gaus_expr}) is centered normal with standard deviation at most $\sigma_x(X)\norm{y}_{\infty}$, where
\begin{equation*}
\sigma_x(X) = \frac{1}{n} \(\sum_{l=1}^n \sum_{i=1}^p \[\sum_{m=1}^n \sum_{j=1}^p b_{lm}x_j X_{jm}\]^2\)^{1/2} = \frac{\sqrt{p}}{n} \(\sum_{l=1}^n \[\sum_{m=1}^n \sum_{j=1}^p b_{lm}x_j X_{jm}\]^2\)^{1/2}.
\end{equation*}
We need to bound this quantity uniformly with respect to all $x$.

\subsection{Concentration} 
Let $x \in \text{Reg}_p$, we estimate $\sigma_x(X)$ using concentration in Gauss space, Lemma \ref{conc_gauss_space}. Due to Jensen's inequality
\begin{align*}
&\mathbb{E}\sigma_x(X) \leq (\mathbb{E}\sigma_x(X)^2)^{1/2} = \frac{\sqrt{p}}{n} \(\sum_{l=1}^n \mathbb{E} \[\sum_{m=1}^n \sum_{j=1}^p b_{lm}x_j X_{jm}\]^2\)^{1/2} \nonumber \\
&= \frac{\sqrt{p}}{n} \(\sum_{l=1}^n \sum_{m=1}^n \sum_{j=1}^p b_{lm}^2x_j^2\)^{1/2} = \frac{\sqrt{p}}{n} \(\sum_{l=1}^n \sum_{m=1}^n b_{lm}^2 \sum_{j=1}^p x_j^2\)^{1/2} = \frac{\sqrt{p}}{n} \norm{B}_{\text{Frob}}.
\end{align*}
Now we consider $\sigma_x: \mathbb{R}^{pn} \to \mathbb{R}$ as a function of the concatenated vector $\vec{X}$ as we did before. We compute the Lipschitz constant with respect to the Euclidean measure on $\mathbb{R}^{pn}$, note that the Euclidean norm on this space coincides with Frobenius norm on the linear space of $p \times n$ matrices. 
\begin{align*}
&\sigma_x(X) = \frac{\sqrt{p}}{n} \(\sum_{l=1}^n \[\sum_{m=1}^n \sum_{j=1}^p b_{lm}x_j X_{jm}\]^2\)^{1/2} = \frac{\sqrt{p}}{n} \(\sum_{l=1}^n \[\sum_{j=1}^p x_j \sum_{m=1}^n b_{lm} X_{jm}\]^2\)^{1/2} \\
&= \frac{\sqrt{p}}{n}\norm{BX^Tx}_2 \leq  \frac{\sqrt{p}}{n} \norm{B} \norm{X} \leq \frac{\sqrt{p}}{n} \norm{B} \norm{X}_{\text{Frob}},
\end{align*}
for the Lipschitz constant
\begin{equation*}
\norm{\sigma_x(X)}_{\text{Lip}} \leq  \frac{\sqrt{p}\norm{B}}{n}.
\end{equation*}
Lemma \ref{conc_gauss_space} now implies that $\forall x \in \text{Reg}_p$ and $t \geq 0$
\begin{equation}
\mathbb{P}\(\sigma_x(X) \geq \frac{\sqrt{p}}{n} \norm{B}_{\text{Frob}} + t\) \leq \frac{1}{2}\exp\(-\frac{t^2n^2}{2p\norm{B}^2}\).
\label{init_bound}
\end{equation}

\subsection{Union bounds}
We return to the estimation of the random variable $(W'x,y)$. Let us fix $u \geq 1$, then $\forall x\in \text{Reg}_p$ we consider the event
\begin{equation*}
\mathcal{E}_x = \left\{\sigma_x(X) \leq \frac{\sqrt{p}}{n} \norm{B}_{\text{Frob}} + u\frac{\sqrt{p}\norm{B}}{n}\right\}.
\end{equation*}
By (\ref{init_bound}) we have
\begin{equation}
\mathbb{P}(\mathcal{E}_x) \geq 1-\frac{1}{2}\exp(-u^2/2).
\label{prob_est}
\end{equation}
Note that $\sigma_x(X)$ and, thus $\mathcal{E}_x$, are independent of $X'$. Let now $x \in \text{Reg}_p(r)$ and $y \in \text{Reg}_p(s)$. As we have observed above, conditioned
on a realization of $X$ satisfying $\mathcal{E}_x$, the random variable $(W'x,y)$ is distributed identically with a centered normal random variable $h$ whose standard deviation is bounded by
\begin{equation*}
\sigma_x(X)\norm{y}_{\infty} \leq \frac{\sqrt{p}}{\sqrt{s}n} \norm{B}_{\text{Frob}} + u\frac{\sqrt{p}\norm{B}}{\sqrt{s}n} =: \sigma.
\end{equation*}
Then by the usual tail estimate for Gaussian random variables, we have
\begin{equation*}
\mathbb{P}((W'x,y) \geq \varepsilon \mid \mathcal{E}_x) \leq \frac{1}{2}\exp(-\varepsilon^2/2\sigma^2).
\end{equation*}
Choose $\varepsilon = u \sigma$ to obtain
\begin{equation*}
\mathbb{P}((W'x,y) \geq \varepsilon \mid \mathcal{E}_x) \leq \frac{1}{2}\exp(-u^2/2), \quad \forall x \in \text{Reg}_p(r), y \in \text{Reg}_p(s).
\end{equation*}
We would like to take the union bound in this estimate over all $y \in \text{Reg}_p(s)$ for a fixed $s$. Note that
\begin{equation}
|\text{Reg}_p(s)| = {p \choose s} 2^s \leq \exp(s \ln{(2ep/s)}),
\label{num_p_b}
\end{equation}
as there are exactly ${p \choose s}$ possibilities to choose the support and $2^s$ ways to choose the signs of the coefficients of a vector in $\text{Reg}_p(s)$, thus
\begin{equation*}
\mathbb{P}\(\max_{y \in \text{Reg}_p(s)}(W'x,y) \geq \varepsilon \mid \mathcal{E}_x\) \leq \frac{1}{2}\exp(s \ln{(2ep/s)}-u^2/2),
\end{equation*}
in order for this bound to be not trivial we assume $u \geq \sqrt{2s \ln{(2ep/s)}}$. Now, using (\ref{prob_est}), we obtain
\begin{align}
&\mathbb{P}\(\max_{y \in \text{Reg}_p(s)}(W'x,y) \geq \varepsilon\) \leq \mathbb{P}\(\max_{y \in \text{Reg}_p(s)}(W'x,y) \geq \varepsilon \mid \mathcal{E}_x\) + \mathbb{P}(\mathcal{E}_x^c) \nonumber \\
&\leq \frac{1}{2}\exp(s \ln{(2ep/s)}-u^2/2) + \frac{1}{2}\exp(-u^2/2) \leq \exp(s \ln{(2ep/s)}-u^2/2).
\label{x_bound}
\end{align}
\subsection{Gathering the bounds}
We continue with formula (\ref{norm_bo}):
\begin{align*}
\mathbb{P}(\norm{W'} \geq t) \leq \mathbb{P}(12 \ceil{\ln{2p}}^2 \max_{x, y \in \text{Reg}_p}(W'x,y) \geq t)  \leq \mathbb{P}(12 \ceil{\ln{2p}}^2 \max_{\substack{r, s \in [p]}} \max_{\substack{x \in \text{Reg}_p(r)\\ y \in \text{Reg}_p(s)}}(W'x,y) \geq t).
\end{align*}
With the help of (\ref{x_bound}) and the bound (\ref{num_p_b}) on the number of points in $\text{Reg}_p(r)$ we obtain
\begin{align*}
\mathbb{P}(\max_{\substack{x \in \text{Reg}_p(r)\\ y \in \text{Reg}_p(s)}}(W'x,y) \geq \varepsilon) \leq \exp(r \ln{(2ep/r)}+s \ln{(2ep/s)}-u^2/2),
\end{align*}
for $u \geq \sqrt{2s \ln{(2ep/s)}}$. 

The function $x \ln{(2ep/x)}$ increases monotonically on the interval $[1,p]$, so $k \ln{(2ep/k)} \leq p \ln{(2ep/p)} = p \ln{(2e)} \leq 2p, \forall k\leq p$. Choose $u \geq 3\sqrt{p}$ to get the bound
\begin{equation*}
\mathbb{P}(\max_{r,s,\in[p]}\max_{\substack{x \in \text{Reg}_p(r)\\ y \in \text{Reg}_p(s)}}(W'x,y) \geq \varepsilon) \leq \exp(-u^2/4).
\end{equation*}
Finally, replace $t$ with $12 \ceil{\ln{2p}}^2 \varepsilon$ to obtain
\begin{align}
\mathbb{P}(\norm{W'} \geq 12 \ceil{\ln{2p}}^2 \varepsilon) \leq\exp(-u^2/4),
\label{pro_b}
\end{align}
where 
\begin{equation}
\epsilon = u\frac{\sqrt{p}}{n} \norm{B}_{\text{Frob}} + u^2\frac{\sqrt{p}\norm{B}}{n}, \quad u \geq 3 \sqrt{p}.
\label{u_eq}
\end{equation}

Integration of (\ref{pro_b}) yields
\begin{equation*}
\mathbb{E}\norm{W'} \leq \frac{12 \ceil{\ln{2p}}^2 \sqrt{p}(4\sigma+\kappa\sqrt{\pi})}{n}.
\end{equation*}
By Lemma \ref{decoup_bound} we obtain
\begin{equation*}
\mathbb{E}\norm{W - W^0} \leq \frac{24 \ceil{\ln{2p}}^2 \sqrt{p}(4\sigma+\kappa\sqrt{\pi})}{n}.
\end{equation*}
Now multiply $W$ by $\Theta^{1/2}$ from left and right to scale the matrices and get the statement of Theorem \ref{main_th_1}.

For Corollary \ref{main_th_2} assume we are given a sequence $\{B_n\}_{n\in\mathcal{S}} \in \mathbb{R}^{n \times n}$such that $\Tr{B_n} = n$, then for every corresponding Wishart matrix
\begin{equation*}
\mathbb{E}\norm{W - W^0} \leq \frac{24 \ceil{\ln{2p}}^2 \sqrt{p}(4\sigma_n+\kappa_n\sqrt{\pi})}{n}.
\end{equation*}
By bounding the values $\kappa_n \leq \kappa$ and $\sigma_n \leq \sigma$ from above we get the desired inequality
\begin{equation*}
\mathbb{E}\norm{W - W^0} \leq \frac{24 \ceil{\ln{2p}}^2 \sqrt{p}(4\sigma+\kappa\sqrt{\pi})}{n}.
\end{equation*}

\section{Acknowledgment}
The author is grateful to Dmitry Trushin for discussions of the proof and useful suggestions.
\bibliographystyle{unsrt}
\bibliography{ilya_bib}

\end{document}